\documentclass[11pt, oneside]{article}

\usepackage{amsbsy,amssymb,amsmath,amsfonts,amsthm,latexsym,multicol,multirow,epsfig,color,subfigure,setspace,array,url,mathrsfs,bm,algpseudocode,algorithm}
\usepackage[stable]{footmisc}
\usepackage[numbers]{natbib}
\usepackage{graphicx}
\usepackage{enumitem}
\usepackage{authblk}
\usepackage{xcolor}
\usepackage{verbatim} 
\usepackage{soul}
\usepackage{lineno}
\usepackage{setspace}

\newtheorem{theorem}{Theorem}[section]
\newtheorem{lemma}{Lemma}[section]

\theoremstyle{definition}
\newtheorem{definition}{Definition}[section]

\newcommand{\X}{\mathbf{X}}
\newcommand{\x}{\bm{x}}
\newcommand{\B}{\bm{\beta}}
\newcommand{\hB}{\widehat{\bm{\beta}}}

\newcommand{\tBR}{\widetilde{\bm{\beta}}_{\mathbf{X}^*}}
\newcommand{\R}{\mathbf{X}^*}
\newcommand{\Y}{\bm{y}}
\newcommand{\hX}{\bm{h}_X}

\newcommand{\Q}{\mathbf{Q}}
\newcommand{\h}{\bm{h}}
\newcommand{\tr}{\text{tr}}
\newcommand{\var}{\text{Var}}
\newcommand{\LL}{\mathbf{L}}
\newcommand{\s}{\bm{l}}
\def\T{\intercal} 

\DeclareMathOperator{\bias}{bias}
\begin{document}
\title{LowCon: A design-based subsampling approach in a misspecified linear model}
\date{}


\author[1]{Cheng Meng}
\author[2]{Rui Xie}
\author[3]{Abhyuday Mandal}
\author[4]{Xinlian Zhang}
\author[3]{Wenxuan Zhong}
\author[3,*]{Ping Ma}

\affil[1]{Institute of Statistics and Big data, Renmin University of China}
\affil[2]{Department of Statistics and Data Science, University of Central Florida}
\affil[3]{Department of Statistics, University of Georgia}
\affil[4]{Division of Biostatistics and Bioinformatics, University of California, San Diego}
\affil[*]{Corresponding author: Ping Ma, pingma@uga.edu}

\maketitle

\doublespacing

\begin{abstract}
We consider a measurement constrained supervised learning problem, that is, (1) full sample of the predictors are given;
(2) the response observations are unavailable and expensive to measure.
Thus, it is ideal to select a subsample of predictor observations,  measure the corresponding responses, and then fit the supervised learning model on the subsample of the predictors and responses.
However, model fitting is a trial and error process, and a postulated model for the data could be misspecified.
Our empirical studies demonstrate that most of the existing subsampling methods have unsatisfactory performances when the models are misspecified.
In this paper, we develop a novel subsampling method, called ``LowCon'', which outperforms the competing methods when the working linear model is misspecified.
Our method uses orthogonal Latin hypercube designs to achieve a robust estimation.
We show that the proposed design-based estimator approximately minimizes the so-called ``worst-case" bias with respect to many possible misspecification terms.
Both the simulated and real-data analyses demonstrate the proposed estimator is more robust than several subsample least squares estimators obtained by state-of-the-art subsampling methods. 
\end{abstract}

\noindent%
\vfill

\newpage
\section{Introduction}\label{sec:intro}
Measurement constrained supervised learning is an emerging problem in machine learning \citep{settles2012active,wang2017computationally,derezinski2018leveraged}.
In this problem, the predictor observations (also called unlabeled data points in machine learning literature) are collected, but the response observations are unavailable and difficult or expensive to obtain.
Considering speech recognition as an example, one may easily get plenty of unlabeled audio data, but the accurate labeling of speech utterances is extremely time-consuming and requires trained linguists. 
For an unlabeled speech of one minute, it can take up to ten minutes for the word-level annotation and nearly seven hours for the phoneme-level annotation \citep{zhu2005semi}.
A more concrete example is the task of predicting the soil functional property, i.e., the property related to a soil's capacity to support essential ecosystem service \citep{hengl2015mapping}.
Suppose one wants to model the relationship between the soil functional property and some predictors that can be easily derived from 
remote sensing data.
To get the response, the accurate measurement of the soil property, a sample of soil from the target area, is needed.
The response thus can be extremely time-consuming or even impractical to obtain, especially when the target area is off the beaten path. Thus, it is ideal to select a subsample of predictor observations,  measure the corresponding responses, and then fit a supervised learning model on the subsample of the predictors and responses.

In this paper, we study the subsampling method and postulate a general linear model for linking the response and predictors.
One of the natural subsampling methods is the uniform subsampling method (also called the simple random subsampling method), i.e., selecting a subsample with the uniform sampling probability. 
For many problems, 
uniform subsampling method performs poorly \citep{cochran2007sampling,thompson2012simple}.
Motivated by the poor performance of uniform sampling, there has been a large number of work dedicated to developing non-uniform random subsampling methods that select a subsample with a data-dependent non-uniform sampling probability \citep{mahoney2011randomized}.
One popular choice of the sampling probability is the normalized statistical leverage scores, leading to the {\it algorithmic leveraging} approach \citep{ma2015leveraging,meng2017effective,zhang2018statistical,ma2020asymptotic}.
Such an approach has already yielded impressive algorithmic and theoretical benefits in linear regression models \citep{mahoney2011randomized,drineas2012fast,ma2015statistical}.
Besides linear models, the idea of {\it algorithmic leveraging} is also widely applied in generalized linear regression \citep{wang2018optimal,ai2019optimal,yu2020optimal}, 
quantile regression \citep{AI2020101512,wang2020optimal}, streaming time series~\citep{xie2019online}, and the Nystr$\ddot{\mbox{o}}$m method \citep{alaoui2015fast}.

Different from random subsampling methods, there also exist some deterministic subsampling methods which select the subsample based on certain rules, especially optimality criteria developed in the design of experiments \citep{pukelsheim2006optimal}, e.g., $A$-, $D$- and $E$-optimality. 
\citet{wang2017computationally} proposed a computationally tractable subsampling approach based on the $A$-optimality criterion.
$D$-optimality criterion was considered in \citet{wang2018information}. 

\begin{figure}[ht]
\begin{center}
\begin{tabular}{ccc}
\includegraphics[scale = .8]{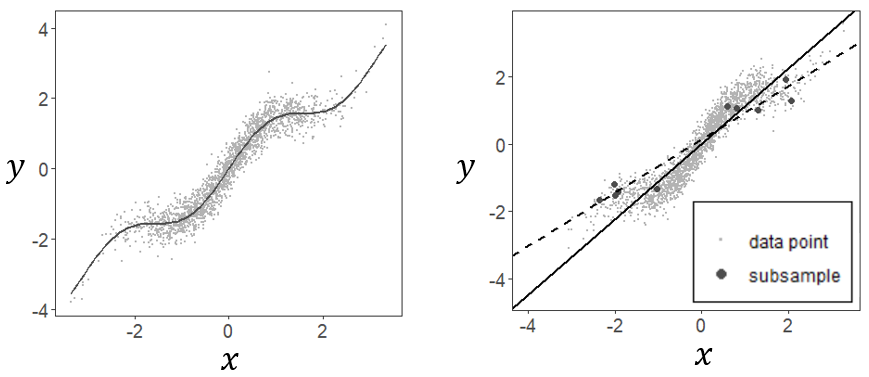}
\end{tabular}
  \caption{The data (gray dots) are generated from a partial-linear model (gray curve). When the non-linear term is omitted, the fitted line (dashed line) based on a leveraging subsample (black dots) deviates severely from the full-sample least squares regression line (solid line). }\label{plot0}
  \end{center}
\end{figure}

While the existing subsampling methods have already shown extraordinary performance on coefficient estimation and model prediction, their performance highly relies on the model specification.
However, the model specification is a trial and error process, during which a postulated model could be misspecified.  When the model is misspecified,  most subsampling methods may lead to unacceptable results. We now demonstrate the issue of model misspecification using a toy example. 
In this example, data are generated from the model $
y_i=x_i+\sin(x_i^2)/2+\epsilon_i$, $i=1,2,\ldots,n$, where $\{\epsilon_i\}_{i=1}^n$ are the i.i.d. standard normal errors. 
In Figure~\ref{plot0}, the data points (gray points) and the true function (the gray curve) are shown in the left panel. 
The right panel shows the full-sample linear regression line (the solid line) based on  $x_i$ only, without the
nonlinear term.
We postulate a linear model without the nonlinear term and randomly select a subsample of size ten (black dots) using the leverage subsampling method \citep{ma2015statistical}.
The subsample linear regression line (the dashed line)  deviates severely from the solid line. 
Such an observation suggests that the performance of a subsample least squares estimator may deteriorate significantly when the model is misspecified.
The poor performance under model misspecifications is not unique to random subsampling methods.  
The success of different deterministic subsampling methods depends on the optimality criteria being used.
The optimality criteria, however, differ from model to model. 
An optimality criterion derived from a postulated model does not necessarily lead to a decent subsampling method for the true model.
We provide more discussion of this example in the Supplementary Material.

In practice, the true underlying model is almost always unknown to practitioners. 
The subsampling hence is highly desirable to be robust to possible model misspecification. 
To achieve the goal, \citet{tsao2012subsampling} proposed to construct a robust estimator using bootstrap. 
One limitation of this method is that it can not be applied to the measurement-constrained setting since the response value for every predictor is needed in this method to compute the estimator.
Another related approach is \citet{pena1999fast}, which aims to carefully select some observations to generate starting points to compute a robust estimator.  
The literature on subsampling methods that yield robust estimations in the measurement-constrained setting is still meager.

In this paper, we bridge the gap by proposing a statistical analysis of the subsampling method in a linear model containing unknown misspecification.
We do so in the context of coefficient estimation via the least-squares on a subsample taken from the full sample.
Our major theoretical contribution is to provide an analytic framework for evaluating the mean squared error (MSE) of the subsample least squares (SLS) estimator in a misspecified linear model.
Within this framework, we show that it is very easy to construct a ``worst-case" sample and a misspecification term for which an SLS estimator will have an arbitrarily large mean squared error.
We also show that an SLS estimator is robust if the information matrix of the subsample has a relatively low condition number, a traditional concept in numerical linear algebra \citep{trefethen1997numerical}.

Based on these theoretical results, we propose and analyze a novel subsampling algorithm, called ``LowCon".
LowCon is designed to select a subsample, which balances the trade-off between bias and variance,  to yield a robust estimation of coefficients.
This algorithm involves selecting the subsample, which approximates a set of orthogonal Latin hypercube design points \citep{ye1998orthogonal}.
We show the proposed SLS estimator has a finite upper bound of the mean squared error, and it approximately minimizes the ``worst-case" bias, with respect to all the possible misspecification terms.
Our main empirical contribution is to provide a detailed evaluation of the robustness of the SLS estimators on both synthetic and real datasets.
The empirical results indicate the proposed estimator is the only one, among all cutting-edge subsampling methods, that is robust to various types of misspecification terms.

The remainder of the paper is organized as follows. 
We start in Section 2 by introducing the misspecified linear model and deriving the so-called ``worst-case" MSE. 
In Section 3, we present the proposed LowCon subsampling algorithm and its theoretical properties.
We examine the performance of the proposed SLS estimator through extensive simulation and two real-world examples in Sections 4 and  5, respectively. 
Section 6 concludes the paper, and the technical proofs are relegated to the Supplementary Material.

\section{Model Setup}

In this section, we first introduce the linear model that contains unknown misspecification.
We then consider the subsample least squares estimator and derive the mean squared error of these estimators under this model.
We show that an SLS estimator is robust if the information matrix of the selected subsample has a relatively low condition number.

Throughout this paper, $||\cdot||$ represents the Euclidean norm.
Let $\lambda_{min}(\cdot)$ and $\lambda_{max}(\cdot)$ be the smallest and the largest eigenvalue of a matrix, and $\bm{\mu}_{min}(\cdot)$ and $\bm{\mu}_{max}(\cdot)$ be the corresponding eigenvectors, respectively. 
We use $s_1(\cdot)$ and $s_p(\cdot)$ to denote the largest and the smallest non-zero singular value of a matrix with $p$ columns, respectively.

\subsection{Misspecified Linear Model}
Suppose the underlying true model has the form
\begin{eqnarray}\label{mmm}
y_i=\x_i^\T\B_0+u_i, \quad i=1,2,\ldots,n,
\end{eqnarray}
where $y_i$'s are the responses, $\x_i$'s are the predictors, $\B_0\in \mathbb{R}^p$ ($p\ll n$) is the vector of unknown coefficients, 
 the random errors $\{u_i\}_{i=1}^n$ are independently distributed, and $u_i$ follows a non-centered normal distribution $N(h(\x_i),\sigma^2)$, $i=1,\ldots,n$.
Let $\mathcal{X}$ be the design space.
In this paper, we assume that the unknown multivariate function $h$ satisfies
\begin{eqnarray}\label{hhh1}
\underset{\x\in\mathcal{X}}{\max}\frac{|h(\x)|}{||\x||}=\alpha,
\end{eqnarray}
where $\alpha $ is a finite positive constant.
When $\x_i=(x_{i1},\ldots,x_{ip})^\T$ has finite values, some examples of $h$ include $h(\x_i)=\sin(x_{i1})$ and $h(\x_i)=x_{i1}x_{i2}$. 
Let $\Y=(y_1, \ldots, y_n)^{\T}$ be the response vector, $\X=(\x_1, \ldots, \x_n)^{\T}$ be the predictor matrix, and $\hX=(h(\x_1),\ldots,h(\x_n))^\T$ be the misspecification term. 
For model identifiability, we assume the matrix $[\X;\hX]$ has a full column rank. Under this assumption, we exclude the case that $h(\x)$ is a linear function of $\x$, i.e., $h(\x_i)$ cannot be a linear combination of $\x_{i1},\ldots,\x_{ip}$.


We consider the scenario that practitioners have no prior information on the true model (\ref{mmm}) and  postulate a classical linear model, 
\begin{eqnarray}\label{mmm2}
y_i=\x_i^\T\B_0+\epsilon_i, \quad i=1,2,\ldots,n,
\end{eqnarray}
where the random errors $\{\epsilon_i\}_{i=1}^n$  are i.i.d. and follow a normal distribution with mean zero and constant variance $\sigma^2$, i.e.,  $N(0,\sigma^2)$.
Model (\ref{mmm2}) is thus a misspecified linear model of the true model (\ref{mmm}). 
Fitting model (\ref{mmm2}) without taking into account the model misspecification may result in the degenerated performance 
of the coefficient estimation and model prediction.
For example, the full-sample ordinary least squares (OLS) estimate, known as the best linear unbiased estimate, is a biased estimate of the true coefficient when the model is misspecified \citep{box1959basis}.
More discussion on misspecified linear models can be found in \citet{kiefer1975optimal} and \citet{sacks1978linear}.

In our measurement-constrained setting, practitioners are given the full sample of predictors $\{\x_i\}_{i=1}^n$.
The responses $\{y_i\}_{i=1}^n$ in model (\ref{mmm}), however, are hidden unless explicitly requested.  
Practitioners are then allowed to reveal a subset of $\{y_i\}_{i=1}^n$, denoted by $\Y^*=(y^*_1,\ldots,y^*_r)^\T$, where $p< r\ll n$.
The goal is to estimate the true coefficient $\B_0$ using $(\x_i^*, y_i^*)$, where $i = 1, \ldots, r$, and $\x_i^*$ is the corresponding predictor for $y_i^*$. 
A natural estimator for the coefficient $\B_0$ is the subsample least squares estimator \citep{wang2017computationally}, 
\begin{eqnarray*}
\tBR=(\mathbf{X}^{*\T}\mathbf{X}^*)^{-1}\mathbf{X}^{*\T}\Y^*,
\end{eqnarray*}
where $\R=(\bm{x}_1^*,\ldots,\bm{x}_r^*)^\T$.
We derive the mean squared error (MSE) and the worst-case MSE of this estimator, in the next subsection.


\subsection{Worst-case MSE}
Let $\Q=(\mathbf{X}^{*\T}\mathbf{X}^*)^{-1}\mathbf{X}^{*\T}$ and $\h=(h(\x^*_1),\ldots,h(\x^*_r))^\T \in \mathbb{R}^r$. 
The MSE of the estimator $\tBR$ (conditional on $\X$) thus can be decomposed as
\begin{align}
\text{MSE}(\tBR) 
& =\tr(\var(\tBR))+[\bias(\tBR)]^\T[\bias(\tBR)]\nonumber\\
& =\sigma^2\tr[(\mathbf{X}^{*\T}\mathbf{X}^*)^{-1}]+[(\mathbf{X}^{*\T}\mathbf{X}^*)^{-1}\mathbf{X}^{*\T}\h]^\T[(\mathbf{X}^{*\T}\mathbf{X}^*)^{-1}\mathbf{X}^{*\T}\h]\nonumber\\
& =\sigma^2\tr[(\mathbf{X}^{*\T}\mathbf{X}^*)^{-1}]+\h^\T\Q^\T\Q\h\label{smse},
\end{align}
where the bias term $\h^\T\Q^\T\Q\h$ is associated with the model misspecification.
Note that when the bias term vanishes, $\h_X=\mathbf{0}$, i.e., when the model is correctly specified, minimizing MSE is equivalent to minimizing the variance term $\sigma^2\tr[(\mathbf{X}^{*\T}\mathbf{X}^*)^{-1}]$.
Further discussion following this line of thinking can be found in \citet{wang2017computationally} and \citet{wang2018information}, in which the authors focused on selecting the subsample that minimizes the variance term.
In our setting, where the model is misspecified, however, minimizing the variance term does not necessarily lead to a small MSE. 

Recall that our goal is to select a subsample such that the corresponding SLS estimator is robust to various model misspecification.
Since the misspecification term $\hX$ is unknown to practitioners, a natural and intuitive approach is to find the ``minimax" subsample that minimizes the so-called ``worst-case" MSE, i.e., the maximum value of $\text{MSE}(\tBR)$ with respect to all the possible choices of the misspecification term $\hX$.
The following lemma gives an explicit form of the worst-case MSE; the proof can be found in the Supplementary Material.

\begin{lemma}[Worst-case MSE]\label{lem1}
Under the regularity condition (\ref{hhh1}), the following inequality holds,
\begin{eqnarray}\label{lem_MSE}
\text{MSE}(\tBR) \leq \sigma^2\tr[(\mathbf{X}^{*\T}\mathbf{X}^*)^{-1}]+\alpha^2\frac{\tr(\mathbf{X}^{*\T}\mathbf{X}^*)}{\lambda_{min}(\mathbf{X}^{*\T}\mathbf{X}^*)}.
\end{eqnarray}
The right-hand side of (\ref{lem_MSE}) is called the worst-case MSE.
\end{lemma}

Two conclusions can be made from Lemma \ref{lem1}. 
First, the worst-case MSE of an SLS estimator can be inflated to arbitrarily large values by a very small value of $\lambda_{min}(\mathbf{X}^{*\T}\mathbf{X}^*)$.
It is thus very easy to construct a ``worst-case" sample and a misspecification term for which an SLS estimator will have unacceptable performance.
Second, $\tBR$ is the most robust SLS estimator if the selected subsample minimizes the worst-case MSE.
Such a subsample, however, is impossible to obtain in real practice, since both values of $\sigma^2$ and $\alpha^2$ are unknown to practitioners.



In this paper, we are more interested in the setting where the misspecified term $h(\x)$ is large enough. In particular, the value of $\alpha^2$ is large enough such that, on the right-hand side of the Inequality (\ref{lem_MSE}), the second term dominates the first term. 
Under this setting, the desired subsample $\R$ should yield a relatively small value of $\tr(\mathbf{X}^{*\T}\mathbf{X}^*)/\lambda_{min}(\mathbf{X}^{*\T}\mathbf{X}^*)$.
Notice that 
\begin{eqnarray}\label{eqn_con}
\tr(\mathbf{X}^{*\T}\mathbf{X}^*)/\lambda_{min}(\mathbf{X}^{*\T}\mathbf{X}^*)\geq p,
\end{eqnarray}
where the equality holds when the condition number of the subsample information matrix, i.e., $\kappa(\mathbf{X}^{*\T}\mathbf{X}^*)\stackrel{\text{def}}{=}\lambda_{max}(\mathbf{X}^{*\T}\mathbf{X}^*)/\lambda_{min}(\mathbf{X}^{*\T}\mathbf{X}^*)$, takes the minimum value $1$. 
Inequality (\ref{eqn_con}) thus suggests the desired subsample $\R$ is the one with a relatively small value of $\kappa(\mathbf{X}^{*\T}\mathbf{X}^*)$. 

We now give another intuition about how $\kappa(\mathbf{X}^{*\T}\mathbf{X}^*)$ is related to the robustness of the SLS estimator.
\citet{casella1985condition} showed that
\begin{eqnarray*}
\frac{||\delta\hB_{ols}||}{||\hB_{ols}||}=\frac{||\delta{(\X^\T\X)^{-1}\X^\T\Y}||}{||(\X^\T\X)^{-1}\X^\T\Y||}\leq\kappa(\X^\T\X)\frac{||\delta{\X^\T\Y}||}{||\X^\T\Y||},
\end{eqnarray*}
where $\delta\hB_{ols}$ and $\delta{\X^\T\Y}$ are perturbations of $\hB_{ols}$ and ${\X^\T\Y}$ respectively. 
Analogously, one can also show that
\begin{eqnarray}\label{eqn7}
\frac{||\delta\tBR||}{||\tBR||}\leq\kappa(\mathbf{X}^{*\T}\mathbf{X}^*)\frac{||\delta{\mathbf{X}^{*\T}\Y^*}||}{||\mathbf{X}^{*\T}\Y^*||}.
\end{eqnarray}
Inequality (\ref{eqn7}) thus suggests that a smaller value of $\kappa(\mathbf{X}^{*\T}\mathbf{X}^*)$ associates with a more robust estimator $\tBR$.

It is worth noting that, if the subsample matrix $\R$ minimizes the worst-case MSE, it does not necessarily minimize $\kappa(\mathbf{X}^{*\T}\mathbf{X}^*)$ simultaneously since both the value of $\sigma^2$ and $\alpha^2$ are not available in practice.
A robust subsample $\R$ should at least yield a relatively small value of $\kappa(\mathbf{X}^{*\T}\mathbf{X}^*)$ and balance the trade-off between the bias and the variance in the Equation (\ref{smse}).
Following this line of thinking, we propose a novel subsampling algorithm,  the details of which are presented in the next section.


\section{LowCon Algorithm}\label{sec:The}
In this section, we present our main algorithm, called ``Low condition number pursuit" or ``LowCon."
In Section 3.1, we introduce the notion of orthogonal Latin hypercube designs (OLHD) and how these can be used to generate a design matrix $\LL$ such that $\kappa(\LL^\T\LL)$ has a relatively small value.
In Section 3.2, we present the detail of the proposed algorithm, which incorporates the idea of OLHD.
In Section 3.3,  we present the theoretical property of the proposed SLS estimator, which is obtained by the LowCon algorithm.
We show that the proposed estimator has a relatively small upper bound of the MSE.

\subsection{Orthogonal Latin Hypercube Design}\label{subsec:space}
Taking a subsample with some specific characteristics has many similarities to the design of experiments, which aims to place design points in a continuous design space, so that resulting design points have certain properties \citep{wu2011experiments}. The theory and methods in the design of experiments are potentially useful for solving subsampling problems. The fundamental difference between the design of experiments and the subsampling is that, in subsampling, the selected points cannot be freely designed in a continuous space as the design of experiments but must be taken from the given finite sample $\{\bm{x}_i\}_{i=1}^n$. To borrow the strength of the design of experiments, we focus on space-filling designs, which aims to place the design points that cover a continuous design space as uniformly as possible \citep{fang2005design,kleijnen2008design,joseph2016space,meng2020more,wang2020lhd}.
In other words, for any point in the experimental region, space-filling designs have a design point close to it.
We thus propose to round the design point to its nearest neighbor in the sample. 
Details are provided in Section 3.2.

We now introduce a specific space-filling design that is of interest, the Latin hypercube design (LHD) \citep{stein1987large, mckay2000comparison, wang2020lhd2}. 
\begin{definition}[Latin hypercube design]
Given the design space $\mathcal{X}=[-1,1]^p$, $\LL \in \mathbb{R}^{r \times p}$
is called a Latin hypercube design matrix if each column of $\LL$ is a random permutation of $\{\frac{1-r}{r},\frac{3-r}{r},\ldots,\frac{r-1}{r}\}$ \citep{steinberg2006construction}.
\end{definition}

Intuitively, if one divides the design space $[-1,1]^p$ into $r$ equally-sized slices in the $j$th ($j=1, \ldots, p$) dimension, a Latin hypercube design ensures that there is exactly one design point in each slice. 
The left panel of Figure \ref{LHD} shows an example of a set of Latin hypercube design points (black dots).
Although uniformly distributed on the marginal, the Latin hypercube design points do not necessarily spread out in the whole design space.
That is to say, a set of LHD points may not be ``space-filling" enough.
To improve the ``space-filling" property of LHD, various methods have been developed \citep{tang1993orthogonal,park1994optimal,fang2002centered,joseph2008orthogonal}. 
Of particular interest in this paper is the orthogonal Latin hypercube design (OLHD), which achieves the goal by reducing the pairwise correlations of LHD \citep{ye1998orthogonal}; see the right panel of Figure \ref{LHD} for an example. 

\begin{figure}[ht]
\begin{center}
\begin{tabular}{cc}
\includegraphics[scale = .6]{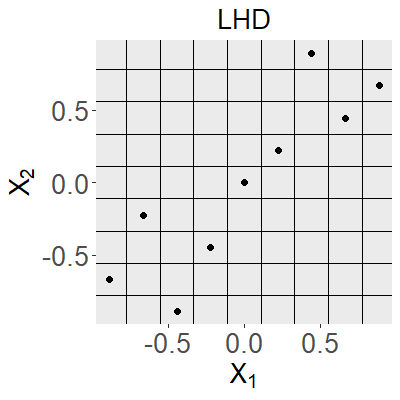}   & \includegraphics[scale = .6]{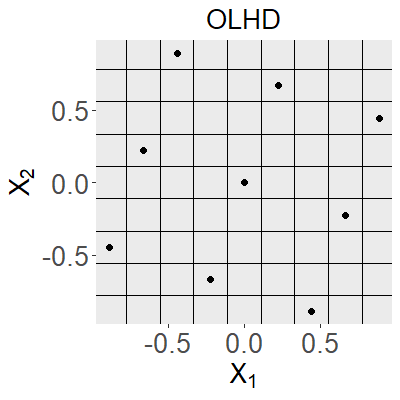}\\
\end{tabular}
  \caption{Example of LHD (left panel) and OLHD (right panel) with nine design points in $[-1,1]^2$. The design points are marked as black dots. As a special case of LHD, OLHD has relatively low pairwise correlation.}\label{LHD}
  \end{center}
\end{figure}

Consider the information matrix $\LL^\T\LL$, where $\LL$ is an OLHD matrix.
Intuitively, the matrix $\LL^\T\LL$ has a relatively small condition number, since all of the diagonal elements of $\LL^\T\LL$ are the same and all of the off-diagonal elements of $\LL^\T\LL$ have relatively small absolute value.
Although there is a lack of theoretical guarantee, empirically, it is known that $\kappa(\LL^\T\LL)$ is in general no greater than 1.13 \citep{cioppa2007efficient}.
Such a fact motivates us to select the subsample that approximates a set of orthogonal Latin hypercube design points.

\subsection{LowCon Subsampling Algorithm}
Without loss of generality, we assume the data points $\{\bm{x}_i\}_{i=1}^n$ are first scaled to $[-1,1]^p$.
The proposed algorithm works as follows.
We first generate a set of orthogonal Latin hypercube design points from a design space $\mathcal{X} \subseteq [-1,1]^p$.
We then search and select the nearest neighbor from the sample for every design point.

The key to success is that the selected subsample can well-represent the set of design points, i.e., each selected subsample point is close-enough to its nearest design point, respectively.
We provide more discussion in Section~3.3 about when such a requirement is met in practice.
Empirically, we find $[-1,1]^p$ may not be a good choice for the design space $\mathcal{X}$.
This is because, in such a scenario, the design points, which are close to the boundary of $[-1,1]^p$, may be too far away from their nearest neighbors, especially when the population density function has a heavy tail.
As a result, a design space that is slightly smaller than $[-1,1]^p$ would be a safer choice. 
We opt to set the design space as $\mathcal{X}_\theta=[\theta_{j1},\theta_{j2}]^p$, where $\theta_{j1}$ and $\theta_{j2}$ are the $\theta$-percentile and $(100-\theta)$-percentile of the $j$th column of the scaled data points, respectively. 
The algorithm is summarized below.

\begin{algorithm}
  \caption{``Low Condition Number Pursuit (LowCon)'' subsampling algorithm}
  \begin{algorithmic}
    \State
    \begin{enumerate}
\item {\textbf{Data normalization:}}
The data points $\{\bm{x}_i\}_{i=1}
^n$ are first scaled to $[-1,1]^p$.
 \item {\textbf{Generate OLHD points:}}
Given the parameter $\theta$ and the design space $\mathcal{X}_\theta \subseteq [-1,1]^p$, generate a set of orthogonal Latin hypercube design points $\{\s_i\}_{i=1}^r$.
\item {\textbf{Nearest neighbor search:}} Select the nearest neighbor for each design point $\s_i$ from $\{\bm{x}_i\}_{i=1}^n$, denoted by $\s^*_i$. The selected subsample is thus given by $\{\s^*_i\}_{i=1}^r$.
\end{enumerate}
  \end{algorithmic}
\end{algorithm}

Figure~\ref{Fig2} illustrates LowCon algorithm. 
The synthetic data points in the left panel were generated from a bivariate normal distribution and are scaled to $[-1,1]^2$. 
A set of orthogonal Latin hypercube design points are then generated, labeled as black triangles in the middle panel. 
For each design point, the nearest data point is selected, marked as black dots in the right panel.
The selected points can well-approximate the design points.

\begin{figure}[ht]
\begin{center}
\includegraphics[scale = .65]{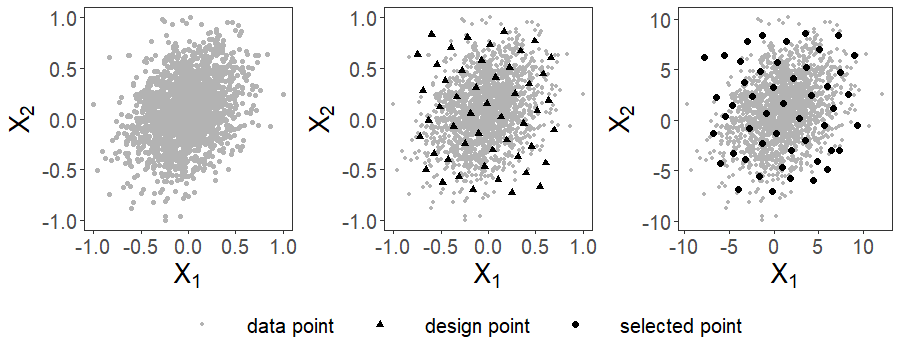}
  \caption{Illustration for Algorithm 1. The data points (gray dots) are first scaled to $[-1,1]^p$, shown in the left panel. A set of OLHD points (black triangles) are generated from $\mathcal{X}_\theta=[-0.8,0.8]^2$, shown in the middle panel. In the left panel, the nearest neighbor for each design point is selected (black dots). }\label{Fig2}
  \end{center}
\end{figure}

Note that the set of design points generated by the orthogonal Latin hypercube design technique is not unique; different sets of design points may result in different subsamples.
Algorithm 1 thus is a random subsampling method instead of a deterministic subsampling method.
In practice, the set of design points $\{\s_i\}_{i=1}^r$ in Algorithm 1 can be randomly generated.

\subsection{Theoretical Results}
We now present the theoretical property of the proposed subsample least squares estimator, obtained by the LowCon algorithm.
Some notations are needed before we show our main theorem. 
Recall that $\LL$ represents an orthogonal Latin hypercube design matrix.
Let $\R_L$ be the subsample matrix obtained by the proposed algorithm. 
One thus can decompose $\R_L$ into a sum of the design matrix $\LL$ and a matrix $\mathbf{D}=(\bm{d}_1,\cdots,\bm{d}_r)^\T$, i.e., $\R_L= \LL+\mathbf{D}$.

Following the notations in Algorithm 1, one can write $\LL=(\s_1,\ldots,\s_r)^\T$ and $\R_L=(\s^*_1,\ldots,\s^*_r)^\T$, where $\s_i$ and $\s^*_i$ represent the $i$th design point and its corresponding nearest neighbor from the sample, respectively.
One thus has $\bm{d}_i=\s^*_i-\s_i$, for $i=1,\ldots,r$.
Intuitively, $\mathbf{D}$ is a random perturbation matrix, and the selected data points can well-approximate the design points if $\mathbf{D}$ is ``negligible".
In such a case, $MSE(\widetilde{\bm{\beta}}_{\R_L})$, which is a function of $\R_L$, can be expanded around $MSE(\widetilde{\bm{\beta}}_{\LL})$ through Taylor expansion.  
From this, we can establish our main theorem below; the proof is relegated to the appendix.

\begin{theorem}\label{mainthm}
Suppose the data follow the model (\ref{mmm}) and the regularity condition (\ref{hhh1}) is satisfied.
Assume $s_p(\LL)>s_1(\mathbf{D})$, where $s_1(\cdot)$ and $s_p(\cdot)$ represent the largest and the smallest singular value of a matrix of $p$ columns, respectively.
A Taylor expansion of $MSE(\widetilde{\bm{\beta}}_{\R_L})$ around the point $\R_L=\LL$ yields the following upper bound,
\begin{align}\label{eqn_9}
MSE(\widetilde{\bm{\beta}}_{\R_L})&\leq \sigma^2p^2\frac{\kappa(\LL^\T\LL)}{tr(\LL^\T\LL)}+\alpha^2p\kappa(\LL^\T\LL)+W.
\end{align}
Here, $W=O(s_1(\mathbf{D}))$ is the Taylor expansion remainder.
\end{theorem}

When the Taylor expansion in Theorem~\ref{mainthm} is valid, three significant conclusions can be made.
First, the theorem indicates that the MSE of the proposed estimator is finite.
Specifically, following the Definition~3.1, we have
$$tr(\LL^\intercal\LL) = \left((\frac{1-r}{r})^2+(\frac{3-r}{r})^2+\ldots+(\frac{r-1}{r})^2\right)\times p.$$
Moreover, the value of $\kappa(\LL^\intercal\LL)$ is in general no greater than 1.13, as discussed in Section~3.1.
Combining these two facts yields an informal but finite upper bound for $MSE(\widetilde{\bm{\beta}}_{\R_L})$, i.e.,
$$MSE(\widetilde{\bm{\beta}}_{\R_L})\leq \sigma^2p^2\frac{1.13}{tr(\LL^\intercal\LL)}+1.13\alpha^2p+W.$$
Recall that Lemma~\ref{lem1} shows that the worst-case MSE of an SLS estimator can be inflated to an arbitrarily large value by a very small value of $\lambda_{min}(\mathbf{X}^{*\T}\mathbf{X}^*)$.
The fact that the proposed estimator has a finite MSE thus indicates the proposed estimator is robust.

Second, the upper bound of the squared bias of the proposed estimator, which equals $\alpha^2p\kappa(\LL^\T\LL)$, is very close to the minimum value of the worst-case squared bias.
This is because the worst-case squared bias has the minimum value of $\alpha^2p$, and the value of $\kappa(\LL^\T\LL)$ is close to 1.
Consider the common situation when the value of $\alpha^2$ is large enough such that, in Inequality~(\ref{lem_MSE}), the bias term dominates the variance term.
Under such a situation, the second conclusion thus indicates the proposed estimator is very close to the ``most robust" estimator, which minimizes the worst-case squared bias. 

Third, the proposed estimator has a finite variance.
Recall that in Algorithm 1, sometimes we may choose a design space $\mathcal{X}_\theta\subset [-1,1]^p$.
The value of $tr(\LL^\T\LL)$ will decrease in such cases, compared to the case when the design space equals $[-1,1]^p$.
The variance of the proposed estimator thus will increase in such cases.
Nevertheless, the variance term will not be inflated to be arbitrarily large, as long as the design space is not too small.
More discussion on the impact of the design space to the Inequality~(\ref{eqn_9}) is relegated to the Supplementary Material.

There are two essential assumptions in Theorem~\ref{mainthm}.
One is that $s_p(\LL)>s_1(\mathbf{D})$, and the other is that the Taylor expansion is valid, i.e., when $s_1(\mathbf{D})$ is ``small".
Although we will evaluate the quality of the proposed estimator empirically in the next section, a precise theoretical characterization of when these two assumptions are valid is currently not available.
Here, we simply give an example such that $s_1(\mathbf{D})$ converges to zero as $n$ goes to infinity, in which case the desired Taylor expansion is apparently valid.
The assumption $s_p(\LL)>s_1(\mathbf{D})$ is also satisfied in such a case, as $n$ goes to infinity, since the value of $s_p(\LL)$ is not relevant to $n$.
Consider the case when the non-zero support of the population distribution is $[-1,1]^p$, i.e., the sample and the design points have the same domain.
In such a case, the distance between each design point and its nearest neighbor converges to zero, as $n$ goes to infinity.
As a result, each entry of the matrix $\mathbf{D}$ converges to zero, and thus $s_1(\mathbf{D})$ converges to zero as well, as $n$ goes to infinity.
Consequently, the desired Taylor expansion is valid in such a case.


\section{Simulation Results}\label{sec4:simul}
To show the effectiveness of the proposed LowCon algorithm in misspecified linear models, we compare it with the existing subsampling methods in terms of MSE.
The subsampling methods considered here are uniform subsampling (UNIF), basic leverage subsampling (BLEV), shrinkage leverage subsampling (SLEV), unweighted-leverage subsampling (LEVUNW) \citep{ma2015statistical,ma2015leveraging}, and information-based optimal subset selection (IBOSS) \citep{wang2018information}. 
The shrinkage parameter for SLEV is set as 0.9, as suggested in \citet{ma2015statistical}.
Through all the experiments in this paper, we set $\theta=1$.
More simulation results with other values of $\theta$ can be found in the Supplementary Material.


We simulate the data from the model (\ref{mmm}) with $n=10^4$, $p=\{10,20\}$ and $r=\{2p,4p,\ldots,10p\}$. 
Three different distributions are used to generate the $\X$ matrix,

{\footnotesize\textbf{D1}}.  $N(\mathbf{1},\boldsymbol{\Sigma})$; 

{\footnotesize\textbf{D2}}.  $0.5 N(\textbf{0}, 2\boldsymbol{\Sigma}) + 0.5 N(\textbf{1}, \boldsymbol{\Sigma})$; 

{\footnotesize\textbf{D3}}.   $t_{10}(\bf{1},\boldsymbol{\Sigma})$,

\noindent where the $(i, j)$th element of $\boldsymbol{\Sigma}$ is set to be $10\times0.6^{|i-j|}$ for $i,j=1,\ldots,p$.
For the coefficient $\bm{\beta}_0$, the first 20\% and the last 20\% entries are set to be 1, and the rest of them are set to be 0.1. 
To show the robustness of the proposed estimator under various misspecification terms, we 
consider five different $h$'s,

{\footnotesize\textbf{H1}}. $h(\x_i)= 0$;

{\footnotesize\textbf{H2}}. $h(\x_i)= 10\sin(x_{i3})$;

{\footnotesize\textbf{H3}}. $h(\x_i)= c_1\cdot x_{i3} x_{i8}$;

{\footnotesize\textbf{H4}}. $h(\x_i)= c_2\cdot x_{i3} \sin(x_{i8})$;

{\footnotesize\textbf{H5}}. $h(\x_i)= c_3\cdot x_{i3}^2$,

\noindent where the constants $c_1$, $c_2$ and $c_3$  are selected so that $\max_{\x\in\{\x_i\}_{i=1}^n}(\{|h(\x)|\}_{i=1}^n)=10$, i.e.,  the response is not dominated by the misspecification term. 
Note that \textbf{H1} does not have any misspecified terms.
Figure~\ref{Fig3} shows the heatmap of the misspecified terms from \textbf{H2} to \textbf{H5}, where $\X$ matrix is generated from \textbf{D1}. 
Only the third and eighth predictors are shown.

\begin{figure}[ht]
\begin{center}
\begin{tabular}{cc}
\includegraphics[scale = .6]{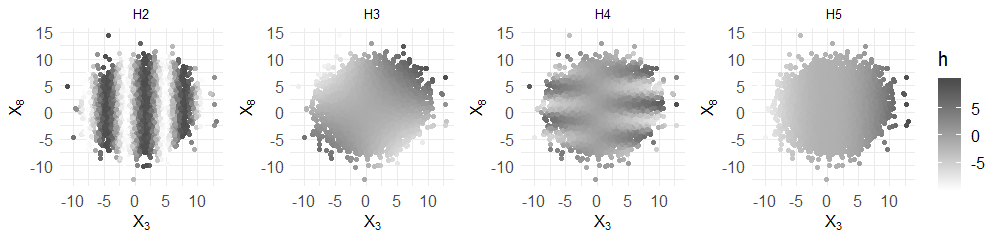}
\end{tabular}
  \caption{The heatmap of ten thousand data points generated from distribution
$\textbf{D1}$ with ten predictors. Only the $3$rd and $8$th predictors are shown. The color demonstrates the values of different model misspecification terms, from $\textbf{H2}$ to $\textbf{H5}$.
}\label{Fig3}
  \end{center}
\end{figure}

We illustrate the subsamples selected by different subsampling methods in Figure~\ref{Fig4}. 
The LEVUNW method is omitted here since the subsample identified by LEVUNW is the same as the subsample identified by BLEV.
The data points (gray dots) are generated from distribution \textbf{D3} with $n=10^4$ and $p=10$, where only the third and the eighth predictors are shown. 
In each panel, a subsample of size 40 is selected (black dots). 
Figure~\ref{Fig4} reveals some interesting facts. 
We first observe the subsamples selected by BLEV and SLEV are more dispersed than the subsample selected by UNIF. 
Such an observation can be attributed to the fact that BLEV and SLEV give more weight to the high-leverage-score data points. 
For the IBOSS method, the selected subsample includes all the ``extreme'' data points from all predictors. 
Such a subsample is most informative when the linear model assumption is valid.
Finally, we observe that the subsample chosen by the proposed LowCon algorithm is most ``uniformly distributed" among all.
Intuitively, such a pattern indicates the selected subsample yields an information matrix that has a relatively small condition number.

\begin{figure}[ht]
\begin{center}
\includegraphics[scale = .41]{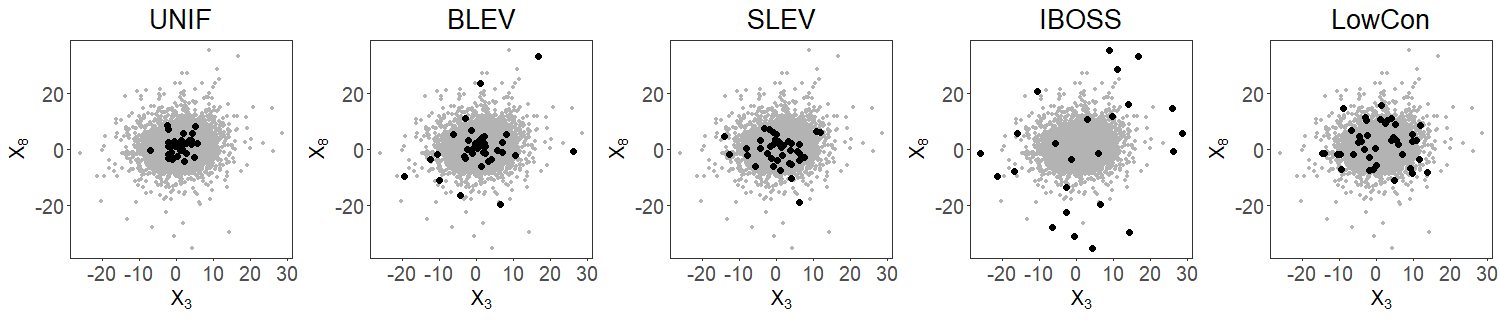}
  \caption{An illustration of five subsamples identified by different subsampling methods. The samples are marked in gray and the selected subsamples are marked in black.}\label{Fig4}
  \end{center}
\end{figure}

To compare the performance for different SLS estimators, we calculate the MSE for each of the SLS estimators based on 100 replicates,
$\mbox{MSE}=\sum_{i=1}^{100}||\bm{\widehat{\B}}^{(i)}-\bm{\beta}_0||^2$/100, where $\bm{\widehat{\B}}^{(i)}$ represents the SLS estimator in the $i$th replication.
Figure~\ref{Fig5} and Figure \ref{Fig6} show the log(MSE) versus different subsample size under various settings, when $p=10$ and $20$, respectively.
In both figures, each row represents a particular data distribution \textbf{D1$-$D3} and each column represents a particular misspecification term {\textbf{H1$-$H5}}. 

\begin{figure}[ht]
  \centering
      \includegraphics[width=16cm]{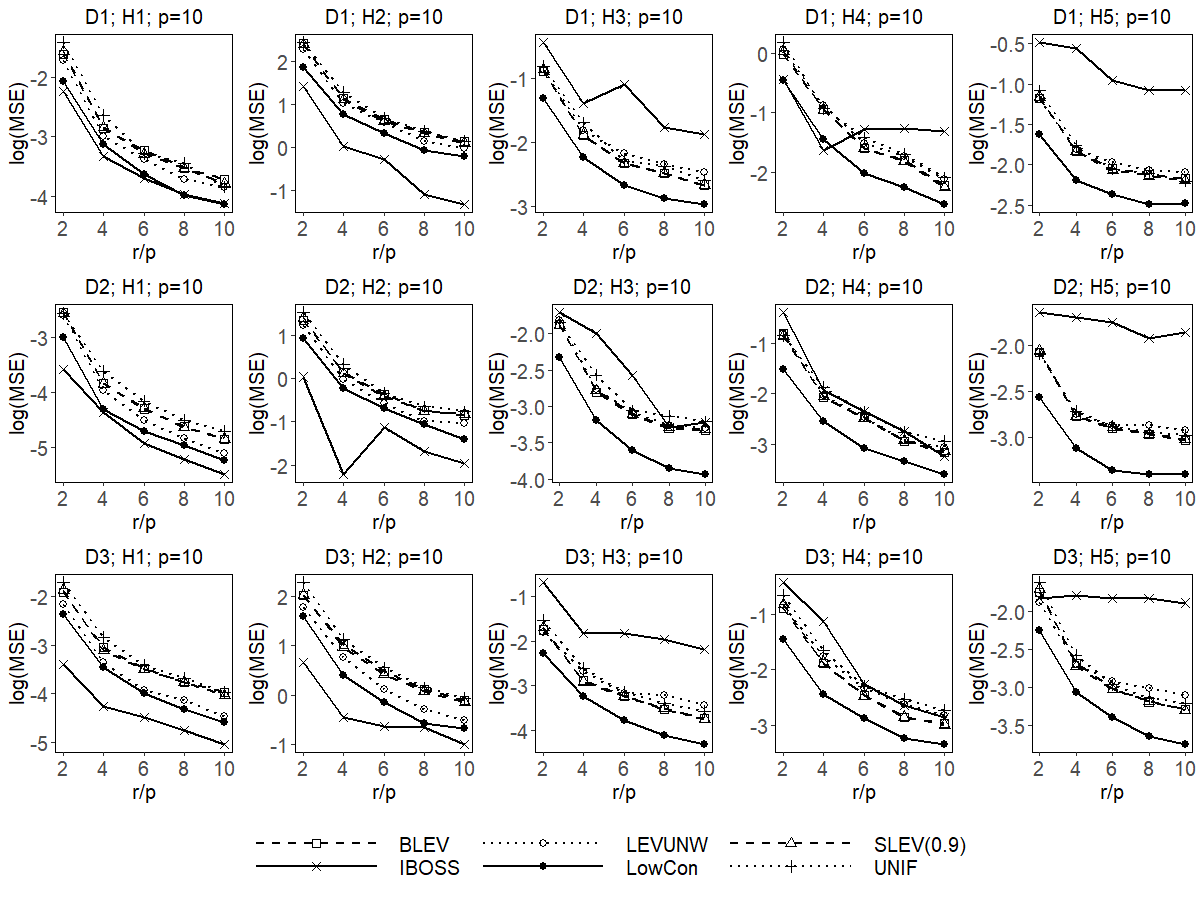}
 \caption{Comparison of different estimators when $p=10$. Each row represents a different data distribution ({\footnotesize\textbf{D1$-$3}}) and each column represents a different misspecification term ({\footnotesize\textbf{H1--5}}). }\label{Fig5}
\end{figure}

In Figures~\ref{Fig5} and \ref{Fig6}, we first observe that UNIF, as expected, does not perform well.
As two random subsampling methods, BLEV and SLEV perform similarly, and both perform better than UNIF in most of the cases.
Such a phenomenon is attributed to the fact that both methods tend to select the data points with high leverage-scores, and these points are more informative for estimating the coefficient, compared to randomly selected points.

Next, we find both LEVUNW and IBOSS have decent performance when the misspecification term equals zero (the leftmost column).
Their performance, however, is inconsistent when the non-zero misspecification term exists, i.e., they perform well in some cases and perform poorly on others.
Note that these two methods, occasionally, are even inferior to the UNIF method.
Such an observation indicates that these two methods are effective when the linear model assumption is correct, but are not robust when the model is misspecified.
We attribute this observation to the fact that the most informative data points derived under the postulated model do not necessarily lead to a decent estimator when the postulated model is incorrect. In fact, the selected subsample can even be misleading and may dramatically deteriorate the performance of the subsample estimator.

\begin{figure}[ht]
  \centering
      \includegraphics[width=16cm]{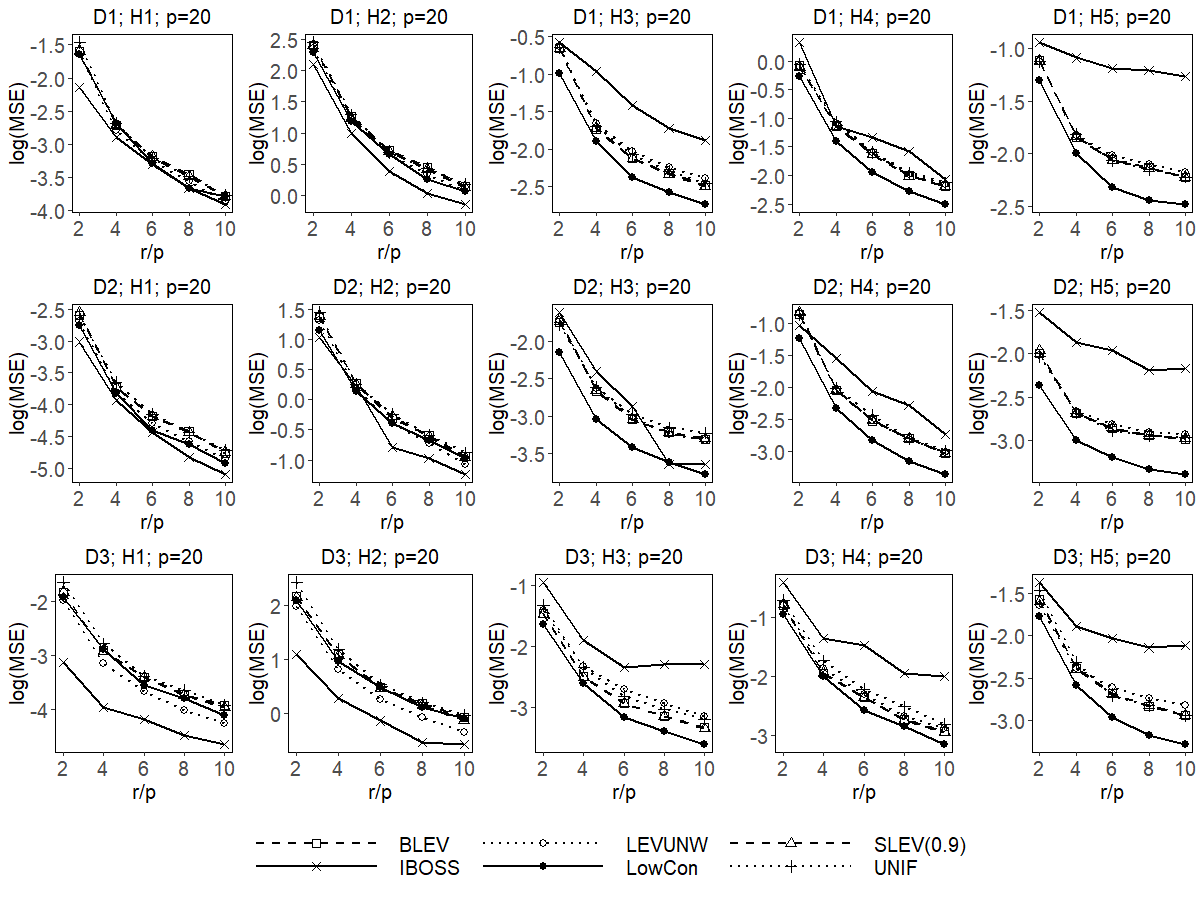}
 \caption{Comparison of different estimators when $p=20$. Each row represents a different data distribution ({\footnotesize\textbf{D1$-$3}}) and each column represents a different misspecification term ({\footnotesize\textbf{H1--5}}). }\label{Fig6}
\end{figure}

Finally, we observe that the proposed LowCon method is consistently better than the UNIF method.
Furthermore, LowCon has a decent performance in most of the cases, especially when the model is misspecified.
This observation indicates LowCon is able to give a robust estimator under various misspecified linear models.
Such success can be attributed to the fact that the proposed estimator has a relatively small upper bound for the worst-case MSE.


\section{Real Data Analysis}\label{sec:real}
We now evaluate the performance of different SLS estimators on two real-world datasets.
One problem in real data analysis is that one does not know the true coefficient.
It is thus impossible to calculate the mean squared error of a coefficient estimate.
To overcome this problem, we consider the full-sample OLS estimator $\widehat{\bm{\beta}}_{OLS}$ and the following three estimators as the surrogates for the true coefficient $\bm{\beta}_0$.
One of them is the M-estimator $\widehat{\bm{\beta}}_{M}$, which is a well-known estimator in robust linear regression \citep{meer1991robust}.
M-estimators can be calculated by using iterated re-weighted least squares, and it is known that such an estimator is more robust to the potential outliers in the data, compared to the OLS estimator \citep{andersen2008modern}. 
We compute the M-estimator using the R package \texttt{MASS} with default parameters.
We also consider the estimator yielded by the cellwise robust M regression method (CRM), denoted by $\widehat{\bm{\beta}}_{CRM}$ \citep{filzmoser2020cellwise}.
Such a method improves the ordinary M-estimator by automatically identifying and replacing the outliers, resulting in a more robust estimator.
We implement the CRM method using the R package \texttt{crmReg}.
The results for the CRM method, however, are omitted in the second dataset, since the code did not stop within a reasonable amount of time.
The last estimator we considered is the cubic smoothing spline estimator for the ``null space" \citep{wahba1990spline,gu2013smoothing,zhang2018smoothing}, denoted by $\widehat{\bm{\beta}}_{SS}$.
We now briefly introduce the cubic smoothing spline estimator in the following.

Suppose the response $y_i$ and the vector of predictors $\x_i=(x_{i1},\ldots,x_{ip})^\T$ are related through the unknown functions $\eta$ such that $y_i=\eta(\x_i)+\epsilon_i$, where $\epsilon_i\stackrel{iid}{\sim}N(0,\sigma^2)$. 
A widely used approach for estimating $\eta$ is via minimizing the penalized likelihood function, 
\begin{equation}\label{eqn:pls2} 
           \frac{1}{n}\sum_{i=1}^n \left(y_i-\eta(\x_i)\right)^2+\lambda J(\eta), 
              \end{equation}
where $\lambda$ is the tuning parameter and $J(\eta)$ is a penalty term. 
We refer to \cite{wahba1990spline,gu2013smoothing,sun2020asympirical} for how to select the tuning parameter and how to construct the penalty term. 
The standard formulation of cubic smoothing splines performs the minimization of (\ref{eqn:pls2}) in a reproducing kernel Hilbert space $\mathcal{H}$.
In this case, the well-known representer theorem \citep{wahba1990spline} states that there exist vectors $\bm{\beta}=(\beta_1,\ldots,\beta_p)^\T$ and $\bm{c}=(c_1,\ldots,c_n)^\T$ such that the minimizer of (\ref{eqn:pls2}) is given by
$\eta(\x)=\sum_{j=1}^p \beta_j x_{ij}+\sum_{i=1}^n c_i H(\x_i,\x).$
Here, the bivariate function $H(\cdot,\cdot)$ is related to the reproducing kernel of $\mathcal{H}$, and we refer to \cite{gu2013smoothing} for technical details.
Let $\mathbf{H}$ be an $n\times n$ matrix where the $(i,j)$-th element equals $H(\x_i,\x_j)$. 
By construction of $\mathcal{H}$, one has $J(\eta)=\bm{c}^\T\mathbf{H}\bm{c}$ \citep{gu2013smoothing}.
Solving the minimization problem in (\ref{eqn:pls2}) thus is equivalent to solving
 \begin{equation}\label{eqn:pls4}
(\widehat{\bm \beta}_{SS}, \widehat{\bm{c}})=
\underset{\bm{\beta},\bm c}{\mathrm{argmin}}
  \frac{1}{n} (\bm y- \mathbf X \bm\beta-\mathbf H \bm c)^\T(\bm y- \mathbf X \bm\beta-\mathbf H \bm c)+\lambda \bm{c}^\T\mathbf{H}\bm{c}.
\end{equation}
We could then view the estimated $\widehat{\bm \beta}_{SS}$ in (\ref{eqn:pls4}) as the ``corrected" estimate of the true coefficient $\bm \beta_0$ that takes into consideration the misspecified terms quantified by $\mathbf{H}\widehat{\bm{c}}$.
We calculate such an estimate using the R package \texttt{gss} with the default parameters.

To compare the performance of different SLS estimators, we calculate the empirical MSE (EMSE) through one hundred replicates.
In the $i$th replicate, each subsampling method selects a subsample leading to an SLS estimator $\hB^{(i)}$.
For each of the four full-sample estimators ($\widehat{\bm{\beta}}_{OLS}$, $\widehat{\bm{\beta}}_{M}$, $\widehat{\bm{\beta}}_{CRM}$, and $\widehat{\bm{\beta}}_{SS}$), the corresponding EMSE is then calculated as
\begin{align*}
\mbox{EMSE}_{OLS} =\sum_{i=1}^{100}||\hB^{(i)}-\widehat{\bm{\beta}}_{OLS}||^2/100, \quad
\mbox{EMSE}_{M} =\sum_{i=1}^{100}||\hB^{(i)}-\widehat{\bm{\beta}}_{M}||^2/100,\\
\mbox{EMSE}_{CRM} =\sum_{i=1}^{100}||\hB^{(i)}-\widehat{\bm{\beta}}_{CRM}||^2/100, \quad \mbox{EMSE}_{SS} =\sum_{i=1}^{100}||\hB^{(i)}-\widehat{\bm{\beta}}_{SS}||^2/100.
\end{align*}

We emphasize that none of these full-sample estimators can be regarded as the gold standard.
However, a robust SLS estimator should at least be relatively ``close" to all of these estimators. 
That is to say, intuitively, a robust SLS estimator yields relatively small values of EMSE$_{OLS}$, EMSE$_M$, EMSE$_{CRM}$, and EMSE$_{SS}$.

Throughout this section, we set the parameter $\theta$ for the proposed LowCon method as 1.
We opt to choose the subsample size $r$ as $5p$, $10p$ and $20p$.
The results in this section show that the proposed SLS estimator yields the smallest empirical mean squared error in almost all of the scenarios.

\subsection{Africa Soil Property Prediction}
Soil functional properties refer to the properties related to a soil's capacity to support essential ecosystem services, which include primary productivity, nutrient and water retention, and resistance to soil erosion \citep{hengl2015mapping}.
The soil functional properties are thus important for planning sustainable agricultural intensification and natural resource management.
To measure the soil functional properties in a target area, a natural paradigm is to first collect a sample of soil in this area, then analyze the sample using the technique of diffuse reflectance infrared spectroscopy \citep{shepherd2002development}.
Such a paradigm might be time-consuming or even impractical if the desired sample of soil from the target area is difficult to obtain.
Predicting the soil functional properties is thus a measurement-constrained problem.

With the help of greater availability of Earth remote sensing data, the practitioners are provided new opportunities to predict soil functional properties at unsampled locations.
One of the Earth remote sensing databases is provided by the Shuttle Radar Topography Mission (SRTM), which aims to generate the most complete high-resolution digital topographic database of Earth \citep{farr2007shuttle}.
In this section, we consider the \textit{Africa Soil Property Prediction} dataset, which contains the soil samples from 1157 different areas ($n=1157$).
We aim to analyze the relationship between the sand content, one of the soil functional properties, and the five features ($p=5$) derived from the SRTM data.
The features include compound topographic index calculated from SRTM elevation data (CTI), SRTM elevation data (ELEV), topographic Relief calculated from SRTM elevation data (RELI), mean annual precipitation of average long-term Tropical Rainfall Monitoring Mission data (TMAP), and modified Fournier index of average long-term Tropical Rainfall Monitoring Mission data (TMFI).
We assume the data follow the model,
\begin{eqnarray}\label{CCPP}
y_i=\beta_0+\beta_1CTI_i+\beta_2ELEV_i+\beta_3RELI_i+\beta_4TMAP_i+\beta_5TMFI_i+u_i, \quad i=1,2,\ldots,n,
\end{eqnarray}
where the random errors $u_i$ are i.i.d. and follow a non-centered normal distribution $N(h(\x_i),\sigma^2)$.
Here, $\x_i=(1,CTI_i,ELEV_i,RELI_i,TMAP_i,TMFI_i)^\T$ and $h(\cdot)$ represents a multivariate function that is unknown to the practitioner.
The postulated model is thus a misspecified linear model.
In our measurement-constrained setting, we assume the response vector is hidden unless explicitly requested.
We then estimate the true coefficient of Model~(\ref{CCPP}), i.e., $(\beta_0,\beta_1,\beta_2,\beta_3,\beta_4,\beta_5)^\T$, using subsampling methods.

The subsampling methods considered here are uniform subsampling (UNIF), basic leverage subsampling (BLEV), shrinkage leverage subsampling (SLEV) with parameter $\alpha=0.9$, unweighted-leverage subsampling (LEVUNW) \citep{ma2015statistical,ma2015leveraging}, information-based optimal subset selection (IBOSS) \citep{wang2018information} and the proposed LowCon method. 
Table 1 summarizes the EMSEs for all six SLS estimators, and the best result in each row is in bold letter.
We observe that the proposed LowCon method yields the best result in every row.

\begin{table}[h!]\caption{EMSEs for the \textit{Africa Soil Property Prediction} dataset}\label{table1}
\centering
\begin{tabular}{ ccccccccc }
\hline
 &  & UNIF & BLEV & SLEV & LEVUNW & IBOSS & LowCon\\
\hline
&EMSE$_{OLS}$ & 5.39 & 2.92 & 3.44 & 2.09 & 34.87 & $\mathbf{1.18}$\\
$r=5p$ &EMSE$_{M}$  & 5.38 & 2.97 & 3.50 & 2.13 & 34.07 & $\mathbf{1.17}$\\
&EMSE$_{CRM}$ & 5.32 & 3.01 & 3.52 & 2.20 & 34.98 & $\mathbf{1.30}$\\
&EMSE$_{SS}$ & 9.82 & 6.71 & 7.31 & 5.71 & 43.59 & $\mathbf{4.89}$\\
\hline
&EMSE$_{OLS}$ & 1.34 & 1.13 & 1.37 & 0.88 & 18.62 & $\mathbf{0.48}$\\
$r=10p$ &EMSE$_{M}$  & 1.36 & 1.17 & 1.35 & 0.92 & 17.97 & $\mathbf{0.51}$\\
&EMSE$_{CRM}$ & 1.38 & 1.21 & 1.37 & 1.00 & 18.51 & $\mathbf{0.61}$\\
&EMSE$_{SS}$ & 5.49 & 5.04 & 5.71 & 4.55 & 27.09 & $\mathbf{4.06}$\\
\hline
&EMSE$_{OLS}$ & 0.61 & 0.45 & 0.64 & 0.38 & 2.84 & $\mathbf{0.27}$\\
$r=20p$ &EMSE$_{M}$  & 0.62 & 0.44 & 0.65 & 0.39 & 2.64 & $\mathbf{0.29}$\\
&EMSE$_{CRM}$ & 0.66 & 0.47 & 0.68 & 0.47 & 2.90 & $\mathbf{0.38}$\\
&EMSE$_{SS}$ & 4.68 & 4.71 & 4.72 & 4.25 & 8.30 & $\mathbf{4.01}$\\
\hline
\end{tabular}
\end{table}

\subsection{Diamond Price Prediction}
The second real-data example we consider is the \textit{Diamond Price Prediction} dataset \footnote{The dataset can be downloaded from https://www.kaggle.com/shivam2503/diamonds.}, which contains the prices and the features of around 54,000 diamonds.
Of interest is to analyze the relationship between the price of the diamond, and three continuous features ($p$=3): weight of the diamond (\textit{caret}), total depth percentage (\textit{depth}), and width of top of diamond relative to widest point (\textit{table}).

As the same setting used in Section 5.1, we assume the data follow a misspecified linear model,
\begin{eqnarray*}
y_i=\beta_0+\beta_1caret_i+\beta_2depth_i+\beta_3table_i+u_i, \quad i=1,2,\ldots,n.
\end{eqnarray*}
Here, the random errors $u_i$ are i.i.d. and follow a non-centered normal distribution $N(h(\x_i),\sigma^2)$, where $\x_i=(1,caret_i,depth_i,table_i)^\T$, and $h(\cdot)$ is a multivariate function that is unknown to the practitioner.
Note that the price of a diamond might be time-consuming or even impossible to obtain if the diamond has not been on the market yet.
We thus assume the value of the response vector is hidden unless explicitly requested, and we estimate the true coefficient using subsampling methods.

Table 2 summarizes the EMSEs for all the subsample estimators, and the best result in each row is in bold letter.
From Table 2, we observe that the proposed LowCon algorithm yields decent performance in all the cases and the best result in most of the cases.

\begin{table}[h!]\caption{EMSEs for the \textit{Diamond Price Prediction} data}\label{table2}
\centering
\begin{tabular}{ ccccccccc }
\hline
& & UNIF & BLEV & SLEV & LEVUNW & IBOSS & LowCon\\
\hline
 &EMSE$_{OLS}$ & 7.01 & 4.24 & 5.29 & 4.67 & 8.96 & $\mathbf{3.40}$\\
$r=5p$ &EMSE$_{M}$  & 7.08 & 4.52 & 5.60 & 4.96 & 6.07 & $\mathbf{4.09}$\\
 &EMSE$_{SS}$ & 11.16 & 9.13 & 10.13 & 9.69 & 10.32 & $\mathbf{8.36}$\\
\hline
 &EMSE$_{OLS}$ & 2.54 & 2.09 & 1.76 & 2.68 & 8.68 & $\mathbf{1.58}$\\
$r=10p$ &EMSE$_{M}$  & 2.88 & 2.37 & $\mathbf{2.15}$ & 2.89 & 5.82 & 2.19\\
 &EMSE$_{SS}$ & 7.41 & 6.83 & 6.53 & 7.54 & 10.18 & $\mathbf{6.28}$\\
\hline
 &EMSE$_{OLS}$ & 1.36 & 0.83 & 1.03 & 1.28 & 8.16 & $\mathbf{0.80}$\\
$r=20p$ &EMSE$_{M}$  & 1.72 & $\mathbf{1.17}$ & 1.40 & 1.32 & 5.38 & 1.33\\
 &EMSE$_{SS}$ & 6.27 & 5.50 & 5.91 & 5.67 & 9.56 & $\mathbf{5.45}$\\
\hline
\end{tabular}
\end{table}


\section{Concluding Remarks}
We considered the problem of estimating the coefficients in a misspecified linear model, under the measurement-constrained setting.
When the model is correctly specified, various subsampling methods have been proposed to solve this problem.
When the model is misspecified, however, we found the worst-case bias for a subsample least squares estimator can be inflated to be arbitrarily large.
To overcome this problem, we aim to find a robust SLS estimator whose variance is bounded, and the worst-case bias is relatively small.
We found such a goal can be achieved by selecting a subsample whose information matrix has a relatively small condition number.
Motivated by this, we proposed the LowCon subsampling algorithm, which utilizes the orthogonal Latin hypercube design to identify sampling points.
We proved the proposed estimator based on the subsample has a finite mean squared error.
Furthermore, the bias of the proposed estimator has an upper bound, which approximately achieves the minimum value of the worst-case bias.
We evaluated the performance of the proposed estimator through extensive simulation and real data analysis.
Consistent with the theorem, the empirical results showed the proposed method has a robust performance.

The proposed algorithm can be easily extended to the cases when the predictor variables are categorical or are a mixture of categorical and continuous variables.
The key idea is to replace the  OLHD in Algorithm 1 by a proper design in a categorical (or mixture) design space.
We refer to \cite{minasny2006conditioned} and the reference therein for more discussion of such designs.
Intuitively, utilizing such designs in Algorithm~1 will result in a subsample in a categorical (or mixture) design space with relatively low ``condition number".


\section*{Acknowledgment}
The authors thank the associate editor and two anonymous reviewers for provided helpful comments on earlier drafts of the manuscript.
The authors would like to acknowledge the  support from the U.S. National Science  Foundation under grants DMS-1903226,  DMS-1925066, the U.S. National Institute of Health under grant R01GM122080.

\clearpage

\bigskip
\begin{center}
{\large\bf SUPPLEMENTARY MATERIALS}
\end{center}

\begin{description}

\item[Title:] Proofs of theoretical results and related materials.

\item[More Discussion of Figure 1:] More discussion on the performance of the subsampling methods in the example shown in Figure 1.

\item[More Simulation Results:] We illustrate the impact of different choices of the parameter $\theta$ through more simulation results.

\item[Proof of Lemma~\ref{lem1}:] A proof of the lemma stating the upper bound of the worst-case MSE.

\item[Proof of Theorem~\ref{mainthm}:] A proof of the main theorem stating the property of the proposed subsample least squares estimator.

\item[More discussion on Theorem 3.1:]
A discussion of the impact of $\theta$ on Theorem~3.1.

\item[R-code:] A package containing code to perform the methods described in the article, the simulation studies, and the real data analysis.

\item[Dataset:] Data are publicly available.

Africa soil property prediction dataset: https://www.kaggle.com/c/afsis-soil-properties/data

Diamond price prediction dataset: https://www.kaggle.com/shivam2503/diamonds

\end{description}

\appendix
\section{More Discussion of Figure 1}
In the example shown in Figure 1, one may wonder about the chances of poor performances of the existing subsampling methods.
To answer this question, we compare the proposed method (LowCon) with the uniform subsampling (UNIF) and the basic leverage subsampling method (BLEV) in terms of estimation error. 
We consider the mean squared error for each of the subsample least squares (SLS) based on one hundred replicates,
$\mbox{MSE}=\sum_{i=1}^{100}||\bm{\widehat{\B}}^{(i)}-\bm{\beta}_0||^2$/100, where $\bm{\widehat{\B}}^{(i)}$ represents the SLS estimator in the $i$th replication.
We consider different subsample sizes $r$ from ten to fifty.
Table 3 summarizes the results, and the best result in each row is in bold letters.
We observe that although the BLEV method performs better than UNIF, it still yields pretty large MSE. 
In other words, both UNIF and BLEV may result in unacceptable performance, especially when $r$ is small.
We also observe the proposed LowCon method yields the best result in every row, indicating the performance of LowCon is robust to the misspecification term. 

\begin{table}[h!]\caption{MSEs for the example in Figure 1}
\centering
\begin{tabular}{ cccc }
\hline
 &  UNIF & BLEV & LowCon \\
\hline
$r=10$ & 0.148 & 0.091 & $\mathbf{0.028}$\\
$r=20$ & 0.118 & 0.075 & $\mathbf{0.027}$\\
$r=30$ & 0.111 & 0.068 & $\mathbf{0.027}$\\
$r=40$ & 0.108 & 0.067 & $\mathbf{0.028}$\\
$r=50$ & 0.105 & 0.060 & $\mathbf{0.028}$\\
\hline
\end{tabular}
\end{table}

\section{More Simulation Results}
Recall that the design space in Algorithm 1 is set to be $\mathcal{X}=[\theta_{j1},\theta_{j2}]^p$, where $\theta_{j1}$ and $\theta_{j2}$ are the $\theta$-percentile and $(100-\theta)$-percentile of the $j$th column of the scaled data points, respectively. 
Throughout Section~4, we set $\theta=1$.
In this section, we illustrate the impact of different choices of the parameter $\theta$.
We consider three different choices of $\theta$, i.e., $\theta=0,5,10$.
Here, $\theta=0$ means the design space $\mathcal{X}=[0,1]^p$.
We let the dimension $p=10$.
Other simulation settings are the same as the ones we used in Section~4.
The results are shown in Figure~\ref{Fig1_new}, Figure~\ref{Fig2_new}, and Figure~\ref{Fig3_new}, respectively.

Consider the cases when $\theta=0$.
First, we observe that LowCon gives the best of the results in most of the cases when the model is correctly-specified, as shown in the leftmost column.
Such an observation is expected since when $\theta=0$, the LowCon method tends to select more data points with large leverage scores, resulting in a better estimation. 
We then observe that the performance of LowCon when $\theta=0$ is not as good as its performance when $\theta=1$, indicating that a positive value of $\theta$ is essential for LowCon to work well in misspecified models. 
Consider the cases when $\theta=10$.
We observe LowCon yields unacceptable performance in many cases. 
Such an observation indicates the choice $\theta=10$ yields a large sampling bias to LowCon, resulting in poor performance.
Finally, when $\theta=5$, we observe LowCon yields reasonably well performance in most of the cases.

In summary, it is essential to select a $\theta$ that is neither too large nor too small for LowCon to perform well in misspecified models.
In practice, we find $\theta\in[0.5,5]$ works reasonably well in most of the cases.

\begin{figure}
  \centering
      \includegraphics[width=16cm]{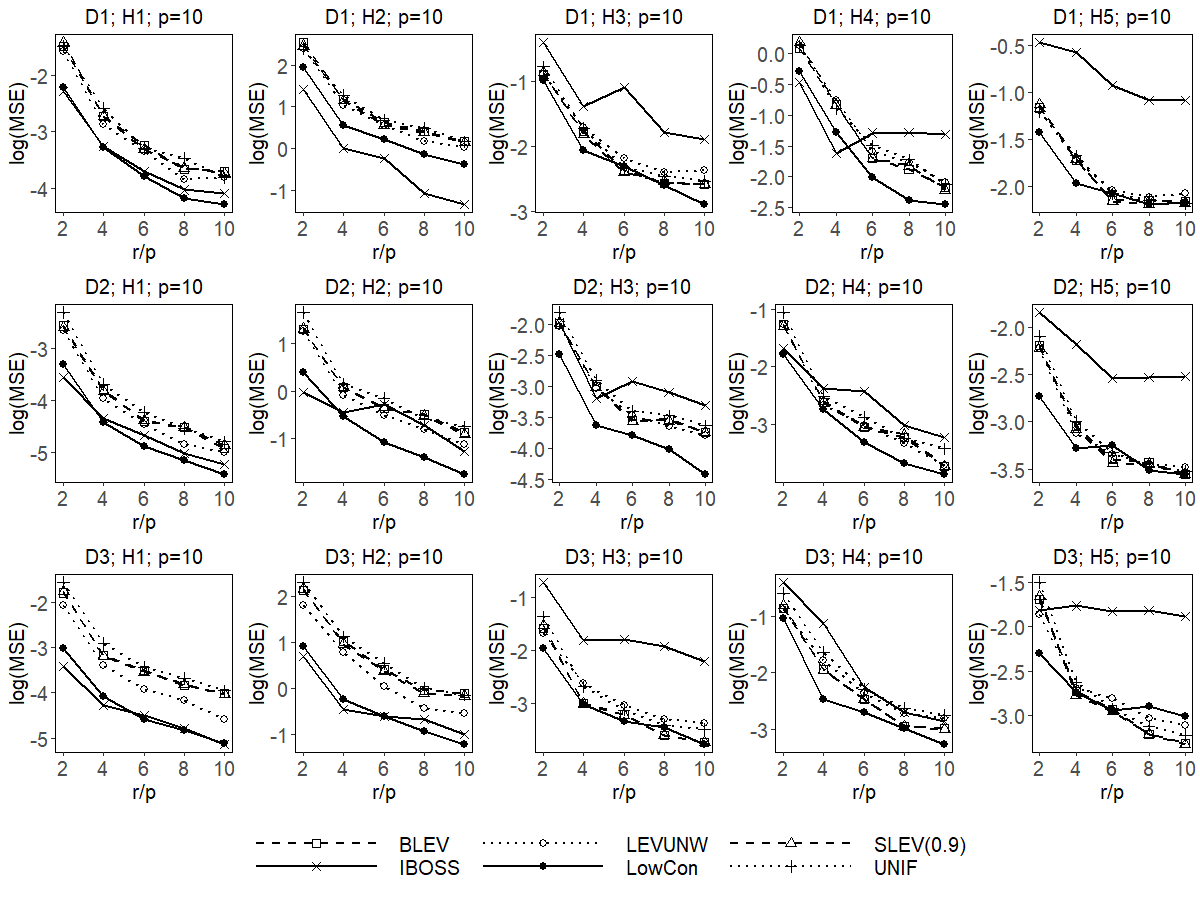}
 \caption{Comparison of different estimators when $p=10$, $\theta=0$. }\label{Fig1_new}
\end{figure}

\begin{figure}
  \centering
      \includegraphics[width=16cm]{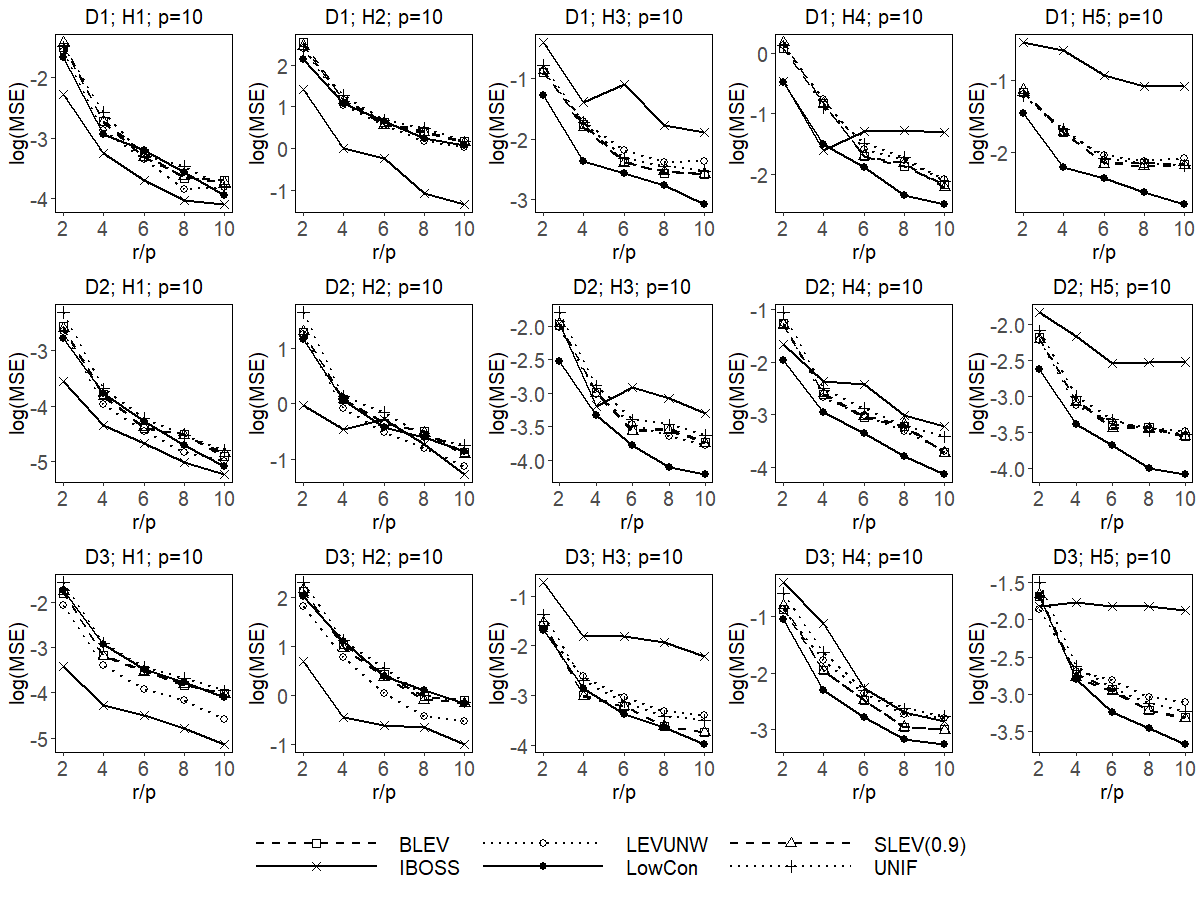}
 \caption{Comparison of different estimators when $p=10$, $\theta=5$. }\label{Fig2_new}
\end{figure}

\begin{figure}
  \centering
      \includegraphics[width=16cm]{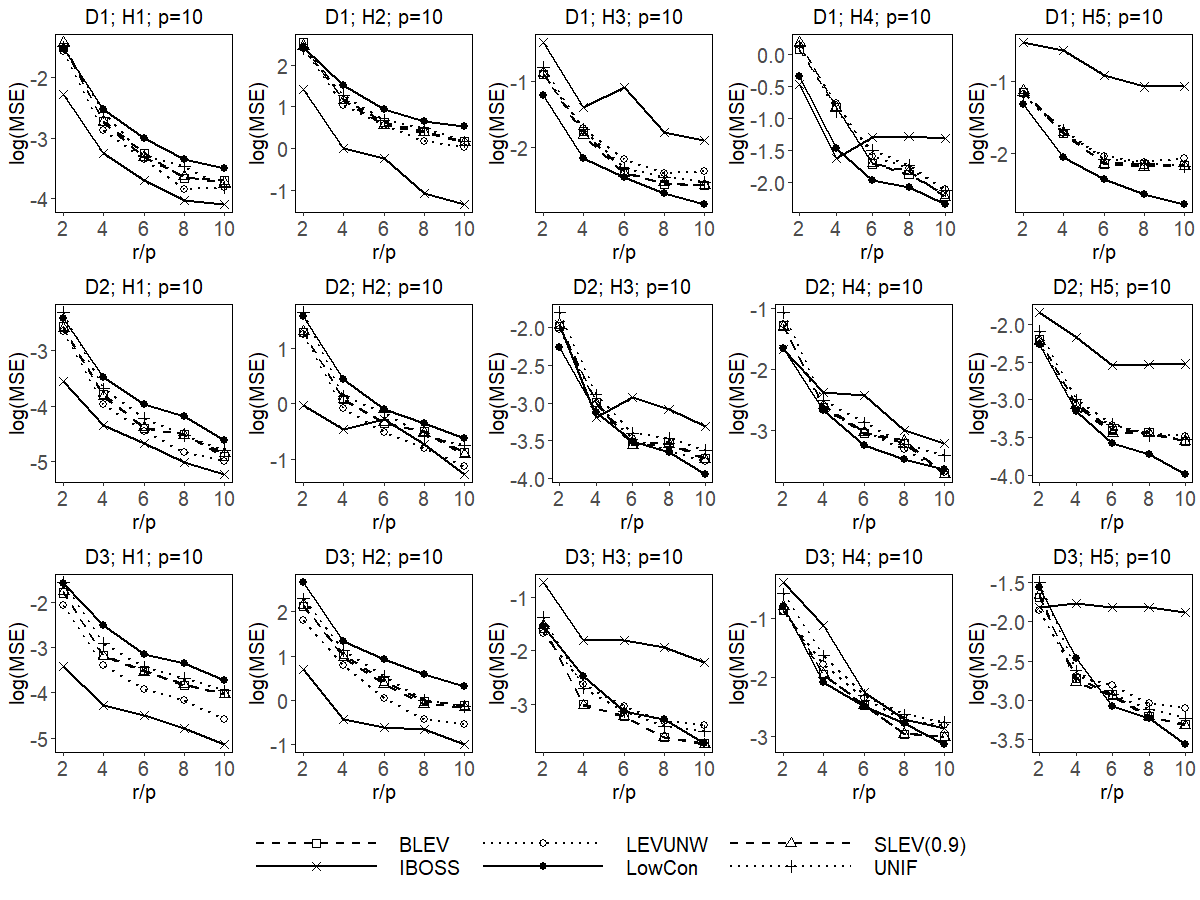}
 \caption{Comparison of different estimators when $p=10$, $\theta=10$. }\label{Fig3_new}
\end{figure}

\section{Proofs of Theoretical Results}

\subsection{Proof of Lemma~\ref{lem1}}
\begin{proof}

Inequality (\ref{hhh1}) yields $||\h||^2 \leq \alpha^2\sum_{i=1}^r||\x^*_i||^2=\alpha^2\tr(\mathbf{X}^{*\T}\mathbf{X}^*).$
One thus has
\begin{align}
\h^\T\Q^\T\Q\h
& \leq \lambda_{max}(\Q^\T\Q)||\h||^2 \label{bias1}\leq \lambda_{max}(\Q^\T\Q) \cdot \alpha^2\tr(\mathbf{X}^{*\T}\mathbf{X}^*)\\
& = \lambda_{max}((\mathbf{X}^{*\T}\mathbf{X}^*)^{-1}) \cdot \alpha^2\tr(\mathbf{X}^{*\T}\mathbf{X}^*) \label{bias3} = \frac{\alpha^2\tr(\mathbf{X}^{*\T}\mathbf{X}^*)}{\lambda_{min}(\mathbf{X}^{*\T}\mathbf{X}^*)} .
\end{align}

Recall that $\bm{\mu}_{max}(\cdot)$ is the corresponding eigenvector to $\lambda_{max}(\cdot)$. 
The first equation in~(\ref{bias1}) holds when $\h=c\cdot\bm{\mu}_{max}(\Q^\T\Q)$ for some real number $c$, and the second equation in~(\ref{bias1}) holds when $||\h||^2=\alpha^2\tr(\mathbf{X}^{*\T}\mathbf{X}^*)$. 
As a result, both equations in~(\ref{bias1}) hold when
$\h=\sqrt{\alpha^2\tr(\mathbf{X}^{*\T}\mathbf{X}^*)}\cdot\bm{\mu}_{max}(\Q^\T\Q).$
The desired result follows directly after plugging Inequality (\ref{bias3}) into Equation (\ref{smse}).
\end{proof}

\subsection{Proof of Theorem~\ref{mainthm}}
\noindent The following Weyl's inequalities are needed in the proof.

\begin{theorem}
\textbf{Weyl's inequalities} \citep{horn1990matrix} Let $\mathbf{A}\in\mathbb{R}^{n\times d}$ and $\mathbf{B}\in\mathbb{R}^{n\times d}$ be two matrices and $t=\min\{n,d\}$. Let $s_1(\mathbf{A}) \geq s_2(\mathbf{A}) \geq \ldots \geq s_t(\mathbf{A}) \geq 0$, $s_1(\mathbf{B}) \geq s_2(\mathbf{B}) \geq \ldots \geq s_t(\mathbf{B}) \geq 0$ and $s_1(\mathbf{A}+\mathbf{B}) \geq s_2(\mathbf{A}+\mathbf{B}) \geq \ldots \geq s_t(\mathbf{A}+\mathbf{B}) \geq 0$ be the singular values of $\mathbf{A}$, $\mathbf{B}$ and $\mathbf{A+B}$, respectively. Then
\begin{eqnarray*}
\left|s_i(\mathbf{A+B})-s_i(\mathbf{A})\right|\leq s_1(\mathbf{B}), \quad i=1,\ldots,t.
\end{eqnarray*}
\end{theorem}

\begin{proof}[Proof of Theorem~\ref{mainthm}]

Let $i=1$; the Weyl's inequalities yield
\begin{eqnarray}\label{eqn_16}
s_1(\R_L)=s_1(\LL+\mathbf{D})\leq s_1(\LL)+s_1(\mathbf{D}).
\end{eqnarray}
Let $i=p$; Weyl's inequalities yield
\begin{eqnarray}\label{eqn_17}
s_p(\R_L)=s_p(\LL+\mathbf{D})\geq s_p(\LL)-s_1(\mathbf{D}).
\end{eqnarray}

Recall that, in Theorem~\ref{mainthm}, we assume $ s_p(\LL)-s_1(\mathbf{D})>0$.
Combining Inequality (\ref{eqn_16}) and Inequality (\ref{eqn_17}) thus yields
\begin{eqnarray}\label{eqn_18}    
\kappa(\mathbf{X}_L^{*\T}\R_L)=\left(\frac{s_1(\R_L)}{ s_p(\R_L)}\right)^2\leq \left(\frac{s_1(\LL)+s_1(\mathbf{D})}{s_p(\LL)-s_1(\mathbf{D})}\right)^2.
\end{eqnarray}

Performing a Taylor expansion of the right-hand side of  Inequality (\ref{eqn_18}), which can be viewed as a function of $s_1(\mathbf{D})$, around the point $0$ yields
\begin{align}
\left(\frac{s_1(\LL)+s_1(\mathbf{D})}{s_p(\LL)-s_1(\mathbf{D})}\right)^2
&=\left(\frac{s_1(\LL)}{s_p(\LL)}\right)^2+2\left(\frac{s_1(\LL)\left(s_1(\LL)+s_p(\LL)\right)}{s_p(\LL)^3}\right)s_1(\mathbf{D})+W_1\nonumber\\
&\leq\kappa(\LL^\T\LL)+4\frac{s_1(\LL)^2}{s_p(\LL)^3}s_1(\mathbf{D})+W_1\nonumber\\
&=\kappa(\LL^\T\LL)+4\frac{\kappa(\LL^\T\LL)}{s_p(\LL)}s_1(\mathbf{D})+W_1\label{eqn_23},
\end{align}
where $W_1=o(s_1(\mathbf{D}))$ is the remainder.
Plugging Inequality (\ref{eqn_23}) back into (\ref{eqn_18}) yields
\begin{eqnarray}
\kappa(\mathbf{X}_L^{*\T}\R_L)\leq\kappa(\LL^\T\LL)+4\frac{\kappa(\LL^\T\LL)}{s_p(\LL)}s_1(\mathbf{D})+W_1.\label{eqn_24}
\end{eqnarray}

We now derive an upper bound for the first term on the right-hand side of Inequality (\ref{smse}).
Note that
\begin{align}
\tr[(\mathbf{X}_L^{*\T}\R_L)^{-1}]&\leq p\lambda_{max}((\mathbf{X}_L^{*\T}\R_L)^{-1})=\frac{p}{s_p(\R_L)^2}\leq\frac{p}{(s_p(\LL)-s_1(\mathbf{D}))^2}\label{eqn_20},
\end{align}
where Inequality (\ref{eqn_17}) is used in the last step.

By performing a Taylor expansion of the right-hand side of Inequality (\ref{eqn_20}) around the point $0$, one has
\begin{align}\label{eqn_21}
\frac{p}{(s_p(\LL)-s_1(\mathbf{D}))^2}&=
\frac{p}{s_p(\LL)^2}+2\frac{\sqrt{p}}{s_p(\LL)^2}s_1(\mathbf{D})+W_2,
\end{align}
where $W_2=o(s_1(\mathbf{D}))$ is the remainder.
Plugging Inequality (\ref{eqn_21}) back into (\ref{eqn_20}) yields
\begin{eqnarray}\label{eqn_22}
\tr[(\mathbf{X}_L^{*\T}\R_L)^{-1}]\leq\frac{p}{s_p(\LL)^2}+2\frac{\sqrt{p}}{s_p(\LL)^2}s_1(\mathbf{D})+W_2.
\end{eqnarray}

Finally, plugging both Inequality (\ref{eqn_24}) and (\ref{eqn_22}) in Inequality (\ref{smse}) yields
\begin{align*}
MSE(\tBR)&\leq \sigma^2\left(\frac{p}{s_p(\LL)^2}+2\frac{\sqrt{p}}{s_p(\LL)^2}s_1(\mathbf{D})+W_1\right)+\alpha^2p\left(\kappa(\LL^\T\LL)+4\frac{\kappa(\LL^\T\LL)}{s_p(\LL)}s_1(\mathbf{D})+W_2\right)\\
&=\left(\frac{\sigma^2}{s_p(\LL)^2}+\alpha^2\kappa(\LL^\T\LL)\right)p+\left(\frac{2\sigma^2\sqrt{p}}{s_p(\LL)^2}+\frac{4\alpha^2p\kappa(\LL^\T\LL)}{s_p(\LL)}\right)s_1(\mathbf{D})+\sigma^2W_1+\alpha^2pW_2\\
&\leq \sigma^2p^2\frac{\kappa(\LL^\T\LL)}{tr(\LL^\T\LL)}+\alpha^2p\kappa(\LL^\T\LL)+O(s_1(\mathbf{D})).
\end{align*}
The fact that $tr(\LL^\T\LL))\leq p\lambda_{max}(\LL^\T\LL))=p\kappa(\LL^\T\LL))s_p(\LL)^2$ is used in the last step.
This completes the proof.

\section{More discussion on Theorem 3.1}
We now discuss the impact of $\theta$ on Theorem~3.1.
Recall that we have $\mathcal{X}_\theta=[\theta_{j1},\theta_{j2}]^p$.
For simplicity, we assume all the marginal distributions of the probability density function are symmetric, i.e., $\theta_{j1}=-\theta_{j2}$, for $j=1,\ldots,p$.
Note that the data are first scaled to $[-1,1]^p$, and thus we have $-1<\theta_{j1}<0<\theta_{j2}<1$.
Let $\LL_\theta$ to denote the design matrix generated from $\mathcal{X}_\theta$.
By the definition of orthogonal Latin hypercube design, it is easy to check that $\kappa(\LL_\theta^\T\LL_\theta) = \kappa(\LL^\T\LL)$.
We also have 
\begin{eqnarray*}
tr(\LL_\theta^\T\LL_\theta) = tr(\LL^\T\LL)\times \prod_{j=1}^p (1-\theta_{j2})^2.
\end{eqnarray*}
In summary, we have 
\begin{align}\label{eqn_10}
MSE(\widetilde{\bm{\beta}}_{\R_L})&\leq
\sigma^2p^2\frac{\kappa(\LL_\theta^\T\LL_\theta)}{tr(\LL_\theta^\T\LL_\theta)}+\alpha^2p\kappa(\LL_\theta^\T\LL_\theta)+W\nonumber\\
&=\sigma^2p^2\frac{\kappa(\LL^\T\LL)}{tr(\LL^\T\LL)\times \prod_{j=1}^p (1-\theta_{j2})^2}+\alpha^2p\kappa(\LL^\T\LL)+W.
\end{align}
Inequality~(\ref{eqn_10}) indicates that a large $\theta$ is associated with a larger upper bound of $MSE(\widetilde{\bm{\beta}}_{\R_L})$.
Furthermore, for fixed $\theta$, a ``heavy-tailed" probability density function also yields a larger upper bound of $MSE(\widetilde{\bm{\beta}}_{\R_L})$.

\end{proof}

\bibliography{ref}
\bibliographystyle{chicago}
\end{document}